%% file: main_ICRA23.tex
\documentclass[letterpaper, 10 pt, conference]{ieeeconf}
\IEEEoverridecommandlockouts

\let\proof\relax
\let\endproof\relax
\usepackage{amsthm}
\usepackage[utf8]{inputenc}
\usepackage{amsmath,amssymb}
\usepackage{booktabs}
\usepackage{graphicx}
\usepackage[font=small]{subcaption}
\usepackage[dvipsnames]{xcolor}
\usepackage{numprint}
\usepackage{csvsimple}
\usepackage{comment}
\usepackage{xcolor}
\usepackage{subcaption}
\usepackage{tikz}
\usepackage[linesnumbered,ruled,noend]{algorithm2e}

\usepackage{tabularx,siunitx}
\usepackage{cite}
\usepackage{rotating} 
\usepackage{hyperref}
\usepackage{eso-pic}
\usepackage{float} 

\newtheorem{theorem}{Theorem}

\newtheorem{lemma}{Lemma}
\newtheorem{definition}{Definition}

\newtheorem{example}{Example}

\DeclareMathOperator*{\reals}{\mathbb{R}}
\newcommand{\N}{\mathcal{N}}

\newcommand{\low}[1]{{#1}^\downarrow}
\newcommand{\up}[1]{{#1}^\uparrow}

\newcommand{\stori}{\text{StoRI}\xspace}
\newcommand{\storm}{\text{StoRM}\xspace}

\newcommand{\btraj}{\mathbf{b}}
\newcommand{\xtraj}{\mathbf{x}}
\newcommand{\ptraj}{\mathbf{b}_t}

\newcommand{\X}{$\mathcal{A}$\xspace}

\title{\LARGE \bf Stochastic Robustness Interval for Motion Planning with Signal Temporal Logic}
\author{Roland B. Ilyes, Qi Heng Ho, and Morteza Lahijanian
\thanks{Authors are with the department of Aerospace Engineering Sciences at the University of Colorado Boulder, CO, USA
        {\tt\small \{\textit{firstname}.\textit{lastname}\}@colorado.edu}}%
}

\begin{document}

      
\maketitle

\input{ICRA23_sections/abstract}
\input{ICRA23_sections/Introduction}

\input{ICRA23_sections/SystemSetup}
\input{ICRA23_sections/StoRI-Extended}
\input{ICRA23_sections/Methodology-Extended}
\input{ICRA23_sections/Experiments-Extended}
\input{ICRA23_sections/Conclusion.tex}

\bibliographystyle{IEEEtran}
\bibliography{bibliography}

\clearpage
\input{ICRA23_sections/Appendix.tex}

\end{document}

%% file: ICRA23_sections/abstract.tex
\begin{abstract}
In this work, we present a novel robustness measure for continuous-time stochastic trajectories with respect to Signal Temporal Logic (STL) specifications. We show the soundness of the measure and develop a monitor for reasoning about partial trajectories. Using this monitor, we introduce an STL sampling-based motion planning algorithm for robots under uncertainty. Given a minimum robustness requirement, this algorithm finds satisfying motion plans; alternatively, the algorithm also optimizes for the measure.  We prove probabilistic completeness and asymptotic optimality of the motion planner with respect to the measure, and demonstrate the effectiveness of our approach on several case studies. 
\end{abstract}

%% file: ICRA23_sections/Introduction.tex
\section{Introduction}
\label{sec:intro}               


In recent years, \emph{Temporal Logics} (TLs) \cite{Lahijanian:AR-CRAS:2017}
have been increasingly employed to formalize complex robotic tasks.  These logics allow precise description of properties over time by combining Boolean logic with temporal operators.  
A popular choice is \emph{Linear TL} \cite{BaierBook2008}, where time is treated as linear and discrete.  
However, robotic systems operate in continuous time, and their tasks often include properties in dense time. Signal TL (STL) \cite{Maler2004} is a variant of TL that  
provides the means for expressing such tasks with evaluations over continuous signals (trajectories).  
For motion planning, however, STL introduces both computational and algorithmic challenges precisely due to the same reason that makes it powerful, i.e., reasoning in continuous time.  
This challenge is exacerbated in real-world robotics, where uncertainty cannot be avoided.  Hence, 
motion planners must reason about the robustness of plans by accounting for uncertainty. This requires a proper notion of a robustness measure.  STL does in fact admit such a measure but only for deterministic trajectories \cite{donze2010robust}. No such a measure is known for stochastic systems.  
This paper takes on this challenge and aims to develop a stochastic robustness measure for STL with the purpose of using it for efficient motion planning with robustness guarantees.

Consider the following mission for the robot in Fig. \ref{fig:form4}: \textit{
Go to the charger in the next 10 minutes. But, if you traverse a puddle of water, stay away from the charger until you dry off on the carpet within 3 minutes}. 
Such a specification is easily captured in STL ($\varphi_3$ in \eqref{formeq:3}). 
Two sample trajectories of an uncertain robot are shown in Fig. \ref{fig:form4}, where the ellipses represent the 90\% confidence contours around the nominal trajectories. Here, both trajectories appear to satisfy the temporal aspects of the specification. However, their spatial robustness is different with regards to their uncertainty. Note that, even though the nominal component of Trajectory 2 gets to the charger without passing through the water puddle, its uncertainty ellipses overlap with the puddle and walls, making it less robust than Trajectory 1 with respect to the task. A capable motion planner must be able to reason about the robustness of these trajectories with respect to both temporal and spatial aspects of the STL task algorithmically.

Much of the literature for control synthesis for STL specifications has been focused on deterministic systems. Success has been found using optimization-based techniques \cite{raman2014,raman2015reactive,OPTSTL1,OPTSTL2}, control barrier function approaches \cite{CBFSTL}, and sampling-based methods \cite{vasile2017sampling,CDCPaper}. 
The sampling-based approaches are particularly suitable for robotics applications given their efficiency and scalability.
None of those methods, however, account for stochasticity in the robot's dynamics. 


\begin{figure}[t]
    \centering
    \includegraphics[width = .84\linewidth,trim={1.8cm 1.5 1.8cm 1.5},clip]{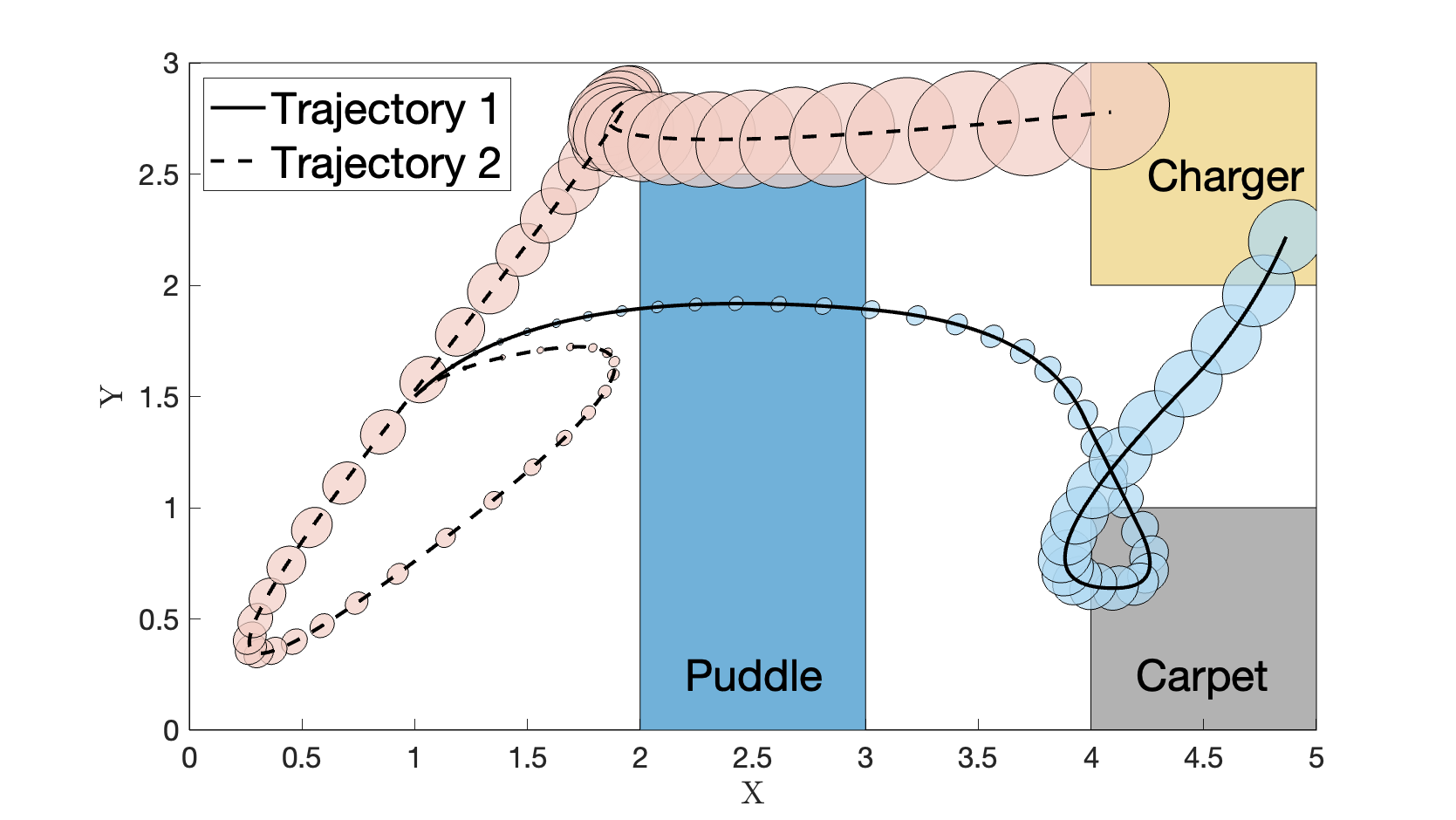}
    \vspace{-1mm}
    \caption{Two trajectories for $\varphi_3$ in \eqref{formeq:3}. Trajectory 1 has a \storm of 0.96 and Trajectory 2 has a \storm of 0.71.}
    \label{fig:form4}
    \vspace{-4mm}
\end{figure}

Reasoning about systems under uncertainty with continuous-time TL specifications is a rapidly developing topic. Recent works \cite{C2TL, farahani2018shrinking, PrSTLBelta, PrSTLSadigh,lee2021signal,tiger2020incremental} introduce new TL as extensions to STL that incorporate uncertainty in the logic itself.
Nevertheless, they do not define a robustness measure for STL. Other works \cite{bartocci2015system,farahani2018shrinking} use STL and reason about the distribution of the STL robustness measure over the realizations of stochastic trajectories. This measure, however, is defined with respect to deterministic trajectories and ignores the knowledge of the system's uncertainty.



In this work, we develop a novel measure of the robustness of continuous-time stochastic trajectories with respect to STL specifications. We refer to it simply as the \textit{Stochastic Robustness Measure} (\storm). \storm allows us to quantify how well a stochastic trajectory satisfies a specification. 
The measure is based on the recursive evaluation of the \textit{Stochastic Robustness Interval} (StoRI) according to the semantics of STL over continuous stochastic trajectories and the probability measure. We also propose a technique to compute the StoRI of partial trajectories, which allows us to reason about (monitor) the satisfaction robustness of the STL properties on an evolving stochastic trajectory.  
Next, we present a sampling-based algorithm that produces motion plans that can (i) satisfy a user-defined bound on \storm, or (ii) asympotically optimize for \storm. We prove the theoretical properties of the algorithm under \storm constraints, and demonstrate the effectiveness of our measure and planner on a variety of STL formulas in several case studies.

In summary, the contributions of this work are four-fold: (i) a novel measure (\storm) to quantify the robustness of stochastic trajectories against STL specifications, (ii) a new monitor to compute the robustness of partial trajectories, (iii) a probabilistically-complete planning algorithm that satisfies STL specifications with \storm constraints for systems under uncertainty and asymptotically optimizes for \storm, and (iv) a series of case studies and benchmarks that reveal properties of \storm and performance of the planner.

%% file: ICRA23_sections/SystemSetup.tex
\section{System Setup}
\label{sec:setup}


Consider a robotic system whose dynamics are described by
a linear stochastic differential equation (SDE):
\begin{equation}
    \label{eq:SDE}
    dx(t) = (A x(t) + B u(t) ) dt + G dw(t),
\end{equation}
where $x \in X \subset \reals^n$ is the state, $u \in U \subset \reals^m$ is the control input,
$w(\cdot)$ is an $r$-dimensional Wiener process (Brownian motion) with diffusion matrix $Q \in \mathbb{R}^{r\times r}$ representing noise, $A \in \reals^{n\times n}$,  $B \in \reals^{n\times m}$, and $G \in \reals^{n\times r}$. 
The initial state of the robot is $x(0) = x_0$. 
The solution to the SDE in \eqref{eq:SDE} is a continuous-time Gaussian process \cite{LinSDE}, 
i.e., 
$$x(t) \sim b_t = \N(\hat{x}(t) , P(t)),$$ 
where $b_t$, called the \textit{belief} of $x(t)$, is a normal distribution with mean $\hat{x}(t) \in X$ and covariance $P(t) \in \reals^{n\times n}$.  The evolution of $b_t$ is governed by:
\begin{align}
    \dot{\hat{x}}(t) &= A\hat{x}(t) + Bu(t),
    \label{eq:mean-dynamics}\\
    \dot{P}(t) &= AP(t) + P(t)A^T + GQG^T 
    \label{eq:covariance-dynamics}
\end{align}
with initial conditions $\hat{x}(0) = x_0$ and $P(0) = 0$.
Note that if, instead of a deterministic initial state, the robot has initial uncertainty described by a Gaussian distribution, only the initial conditions $\hat{x}(0)$ and $P(0)$ change.

Let $T \in \reals_{\geq 0}$ be a time duration.
Then, given
a controller $\mathbf{u}: [0,T] \rightarrow U$,  a \textit{belief (stochastic) trajectory} $\btraj$ over time window $[0,T]$ for the robotic system \eqref{eq:SDE} can be computed (predicted) using \eqref{eq:mean-dynamics} and \eqref{eq:covariance-dynamics}. 
An execution of this controller on the robotic system, called
a \textit{realization} or \textit{sample} of $\btraj$, is a \textit{state trajectory} $\xtraj$ over time duration $[0,T]$.

We are interested in properties of System  \eqref{eq:SDE} with respect to a set of \emph{linear predicates} defined in state space $X$. Let $\mathcal{H} = \{h_1,h_2,...,h_l\}$ be a given set of functions where $h_i: \mathbb{R}^n \rightarrow \mathbb{R}$ is a linear function for every $1 \leq i \leq l$. Then, the set of predicates $M = \{\mu_1,...,\mu_l\}$ is defined on $\mathcal{H}$ such that $\forall i \in \{1,...,l\}$, the Boolean values of $\mu_i:X\rightarrow \{\top,\bot\}$ is determined by the sign of function $h_i$ as:  
\begin{equation}
    \mu_i(x) = 
    \begin{cases}
        \top \text{ if } h_i(x) \geq 0 \\
        \bot \text{ if } h_i(x) < 0.
    \end{cases}
    \label{eq:preddef}
\end{equation}

To express desired properties of the robot with respect to the set of predicates $M$, we use \textit{signal temporal logic} (STL) \cite{Maler2004}. STL is a logic that allows specification of
real-time temporal properties,
and is therefore well-suited for continuous-time systems such as the one in \eqref{eq:SDE}.

\begin{definition}[STL Syntax]
    \label{def:STLSyntax}
    The STL Syntax is recursively defined by:
    \begin{equation*}
        \phi := \top \mid \mu \mid \lnot \phi \mid \phi \land \phi \mid \phi \, \mathcal{U}_I \phi
    \end{equation*}
    where $\mu \in M$, and $I = \langle a,b \rangle$ is a time interval with $\langle \in \{(,[\}$, $\rangle \in \{),]\}$, $a,b \in \mathbb{R}_{\geq 0}$ and $a < b < \infty$. Notations $\top$, $\lnot$, and $\land$ are the Boolean ``true,"  ``negation,'' and ``conjunction,'' respectively, and $U$ denotes the temporal ``until'' operator. 
\end{definition}

\noindent
The temporal operator \textit{eventually} ($\Diamond$) is defined as $\Diamond_{I}\phi \equiv \top \mathcal{U}_{I}\phi$ and the operator \textit{globally} is defined as $\square_{I}\phi \equiv \lnot \Diamond_{I}\lnot \phi$. 

\begin{definition}[STL Semantics]
\label{def:STLSemant}
    The semantics of STL is defined over a state trajectory realization $\xtraj$ 
    at time $t$ as:
        \begin{alignat*}{2}
            (\xtraj,t) &\models \top &&\Longleftrightarrow \top \\
            (\xtraj,t) &\models \mu &&\Longleftrightarrow h(\xtraj(t)) \geq 0 \\
            (\xtraj,t) &\models \lnot \phi &&\Longleftrightarrow (\xtraj,t) \not\models \phi\\
            (\xtraj,t) &\models \phi_1 \land \phi_2 &&\Longleftrightarrow (\xtraj,t) \models \phi_1 \land (\xtraj,t) \models \phi_2\\
            (\xtraj,t) &\models \phi_1 \mathcal{U}_{\langle a,b \rangle} \phi_2 &&\Longleftrightarrow \exists t' \in \langle t+a,t+b \rangle \, s.t. \, (\xtraj,t') \models \phi_2 \\ & && \quad \ \quad \land \forall t'' \in [t,t'], \; (\xtraj,t'') \models \phi_1
        \end{alignat*}
    where $\models$ denotes satisfaction. A state trajectory $\xtraj$ satisfies an STL formula $\phi$ if $(\xtraj,0) \models \phi$. 
\end{definition}

In addition to the Boolean semantics, STL admits \textit{quantitative} semantics \cite{Lahijanian:AR-CRAS:2017}. This is traditionally defined as a robustness metric \cite{donze2010robust} based on Euclidian distance. 
This metric is well-defined for deterministic systems.
For a stochastic process such as System \eqref{eq:SDE}, however, evaluation of the satisfaction of a belief trajectory is not straightforward.  

In this work, we aim to develop an appropriate measure of robustness to evaluate the satisfaction of STL formulas by a belief trajectory. Specifically, we are interested in using this measure for robust motion planning for System \eqref{eq:SDE}. Therefore, the robustness measure must be appropriate for planning, namely sampling-based motion planning algorithms. This entails that the measure must be able to provide useful information for complete as well as partial trajectories. In the next section, we introduce such a measure. Then, we present a motion planning algorithm that asymptotically optimizes for this measure.

%% file: ICRA23_sections/StoRI-Extended.tex
\section{Stochastic Robustness Measure}
\label{sec:stori}

A popular method of measuring robustness of belief trajectories in motion planning is based on the notion of chance constraints \cite{blackmore2006probabilistic,luders2010chance,ho2022gaussian}. Chance constraints require that the probability of constraint violation (e.g., avoiding obstacles) not exceed some prescribed value. Motion planners typically enforce this chance constraint at each time step. 

Such an approach is difficult to extend to STL formulas. This is primarily due to temporal operators and their time intervals,
e.g., the probability of violating $\Diamond_I \phi$ only needs to be below a prescribed value at one time step in $I$, but it is not clear at which time step to enforce this. 
Some approaches find success by generating constraints over time windows, and allocating risk of violating the constraints among time steps~\cite{C2TL,farahani2018shrinking}. However, those methods are limited to reasoning about discrete-time trajectories. This is in conflict with the fundamental idea of STL, which is expressing properties over real-valued continuous-time intervals.

Furthermore, those approaches only provide a \textit{qualitative} (boolean) judgement of a trajectory's satisfaction with respect to a chance constraint. In contrast, we seek to define a measure that provides a \textit{quantitative} judgement of a trajectory's satisfaction. This is analogous to robustness for deterministic STL that provides a quantitative measure beyond the qualitative Boolean semantics. Such a measure can then be extended to reason about partial trajectories, and hence, is advantageous for iterative methods for planning such as sampling-based algorithms. To illustrate what such a measure might convey, consider the following example.
\begin{example}
    \label{ex:intuition}
    \emph{
    Consider the formulas $\phi_1 = \square_I \mu$ and $\phi_2 = \Diamond_I \mu$. In the case of the $\square$ operator, which states that a property must hold for all time $ t \in I$, a quantitative measure of robustness
    could be characterized by the point with the lowest probability of satisfying $\mu$ in $I$, i.e., $\min_{t\in I} P(h(\xtraj(t)) \geq 0)$, where $h$ is the linear function that $\mu$ is defined on. If the probability of violation at that point is below a certain threshold, then it is also below that threshold at every other point. Similarly, for the $\Diamond$ operator, which states that a property must hold for \textit{a} time point $ t \in I$, we could look at the point with the \textit{highest} probability of satisfying $\mu$, i.e., $\max_{t\in\langle a,b \rangle} P(h(\xtraj(t)) \geq 0)$. If the probability of violation at that point is above a certain threshold, then formula $\phi_2$ is satisfied. A quantitative measure of stochastic robustness must incorporate this conflicting treatment of temporal operators.
    }
\end{example}
This intuition guides the development of the Stochastic Robustness Interval (\stori) as defined below.

\begin{definition}[Stochastic Robustness Interval]
    \label{def:stori}
    The Stochastic Robustness Interval (\stori) of a belief trajectory $\btraj$ over time window $[0,T]$ with respect to an STL formula $\phi$ is a functional $f(\phi, \btraj)$:
    \begin{equation*}
        f(\phi, \btraj) = [\low{f}(\phi, \btraj), \up{f}(\phi,  \btraj)]
    \end{equation*}
    such that $0 \leq \low{f}(\phi,\btraj) \leq \up{f} (\phi,\btraj) \leq 1$.  
    For a $t \in [0,T]$,
    let $\btraj^t$ be the time-shifted suffix of trajectory $\btraj$ such that 
    for all $t' \in [0,T-t]$, $\btraj^t(t') = \btraj(t+t')$. 
    Then, lower bound $\low{f}(\phi,\btraj)$ and upper bound $\up{f}(\phi,\btraj)$ are recursively defined by:
    \begin{align*}
        &\low{f}(\top,\btraj) = \up{f}(\top, \btraj) = 1, \\
        &\low{f}(\mu, \btraj) = \up{f}(\mu, \btraj) = P(h(x(0)) \geq 0 ), \\
        &\low{f}(\lnot \phi, \btraj) = 1 - \up{f}(\phi,\btraj), \\
        &\up{f}(\lnot \phi, \btraj) = 1 - \low{f}(\phi,\btraj), \\
        \end{align*} 
        \begin{align*}
        &\low{f}(\phi_1 \land \phi_2, \btraj) = \max\big\{\low{f}(\phi_1,\btraj) + \low{f}(\phi_2,\btraj) -1,0\big\}, \\
        &\up{f}(\phi_1 \land \phi_2, \btraj) = \min\big\{\up{f}(\phi_1,\btraj), \up{f}(\phi_2,\btraj)\}, \\
        &\low{f}(\phi_1 U_{\langle a,b \rangle} \phi_2, \btraj) = \\
        & \quad \max_{t\in\langle a,b \rangle} \Big\{ \max \big\{\low{f}(\phi_2,\btraj^t) + \min_{t'\in[0,t]}\low{f}(\phi_1,\btraj^{t'}) - 1,0\big\} \Big\}, \\
        &\up{f}(\phi_1 U_{\langle a,b \rangle} \phi_2, \btraj) = \\
        & \quad \max_{t\in\langle a,b \rangle} \Big\{ \min \big\{\up{f}(\phi_2,\btraj^t), \min_{t'\in[0,t]}\up{f}(\phi_1,\btraj^{t'})\big\}\Big\}.
    \end{align*}
\end{definition}

We define the Stochastic Robustness Measure (\storm) to be the lower bound of \stori. 

\begin{definition}[Stochastic Robustness Measure]
    \label{def:StochRo}
    The \emph{Stochastic Robustness Measure} (\storm) of a belief trajectory $\btraj$ with respect to STL formula $\phi$ is the lower bound of the \stori, i.e, $\low{f}(\phi, \btraj)$.
\end{definition}

The \stori is an interval that aims to quantify how robustly a belief trajectory $\btraj$ satisfies an STL formula $\phi$. The \storm of a belief trajectory is the lower bound of its \stori. 
A trajectory always satisfies a Boolean $\top$, so the \stori for $\phi = \top$ is $[1,1]$ for every trajectory. If $\phi = \mu$ is a linear predicate, both bounds of the \stori are the probability that the state at time zero $x(0)$
satisfies the linear predicate. Works \cite{blackmore2006probabilistic} \cite{pairet2021} outline an efficient way to calculate this probability for Gaussian distributions. However, we stress that the \stori is defined for general belief distributions, and also applies to non-Gaussian beliefs. 

The \stori of the negation of a formula $\lnot \phi$ derives from the Unit Measure Axiom of Probability \cite{ProbBook}. Note that the lower bound of the \stori of $\lnot\phi$ depends on the upper bound of the \stori of $\phi$, and vice-versa.

\begin{figure*}
\centering
  \includegraphics[width=0.32\linewidth,trim={2cm 0 1.3cm 0},clip]{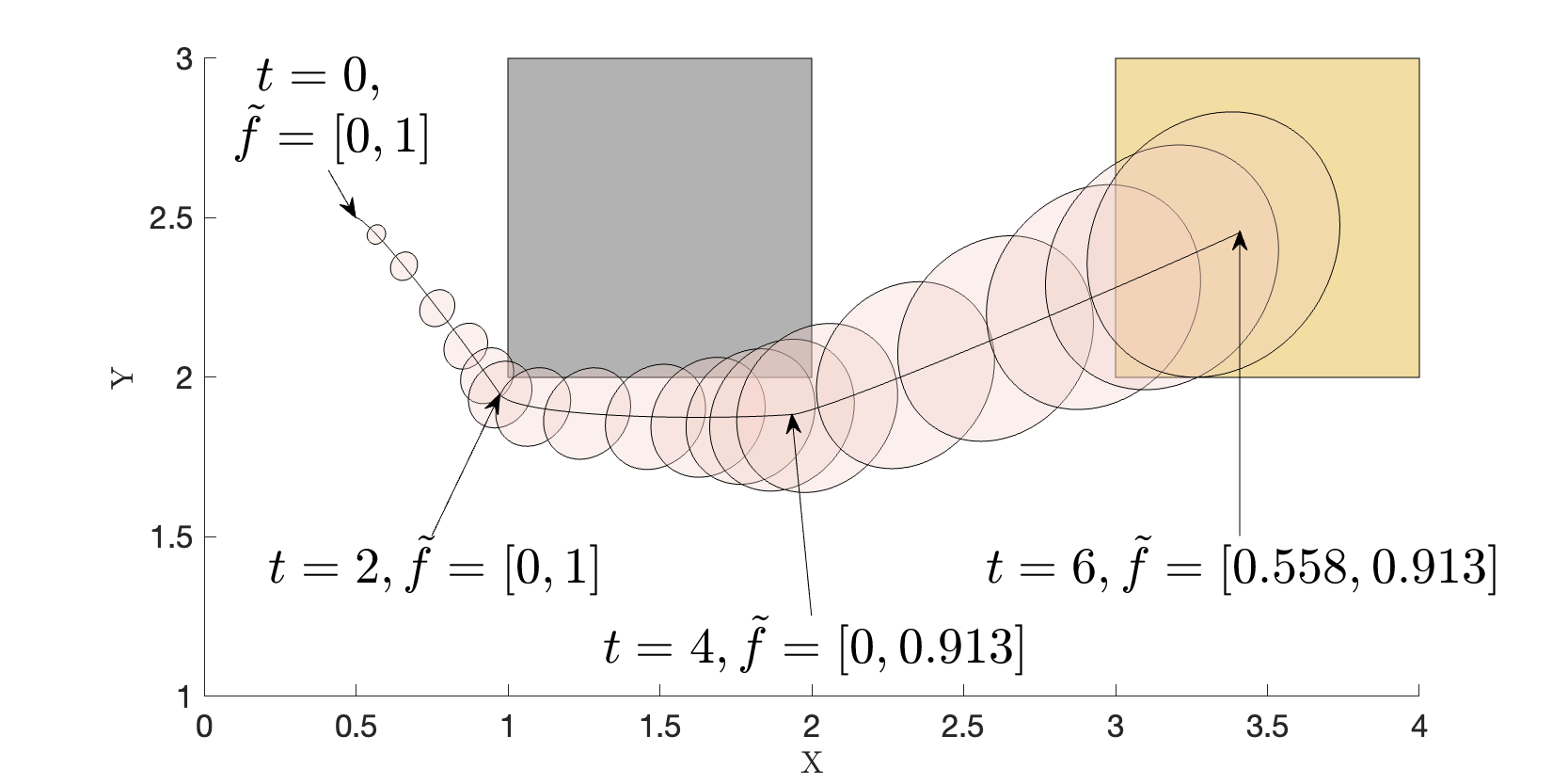}
  \hfill
  \includegraphics[width=0.32\linewidth,trim={2cm 0 1.3cm 0},clip]{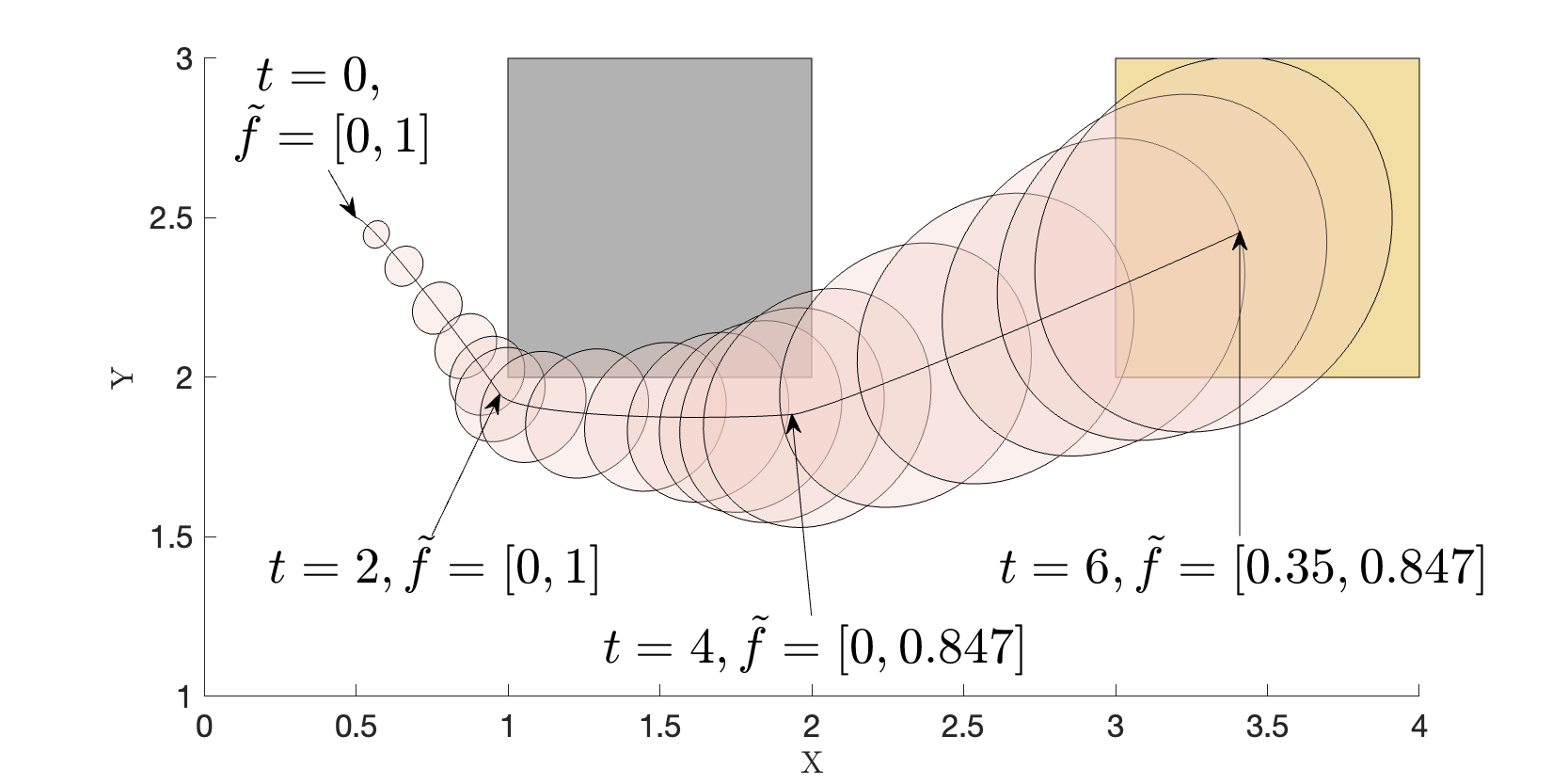}
  \hfill
  \includegraphics[width=0.32\linewidth,trim={2cm 0 1.3cm 0},clip]{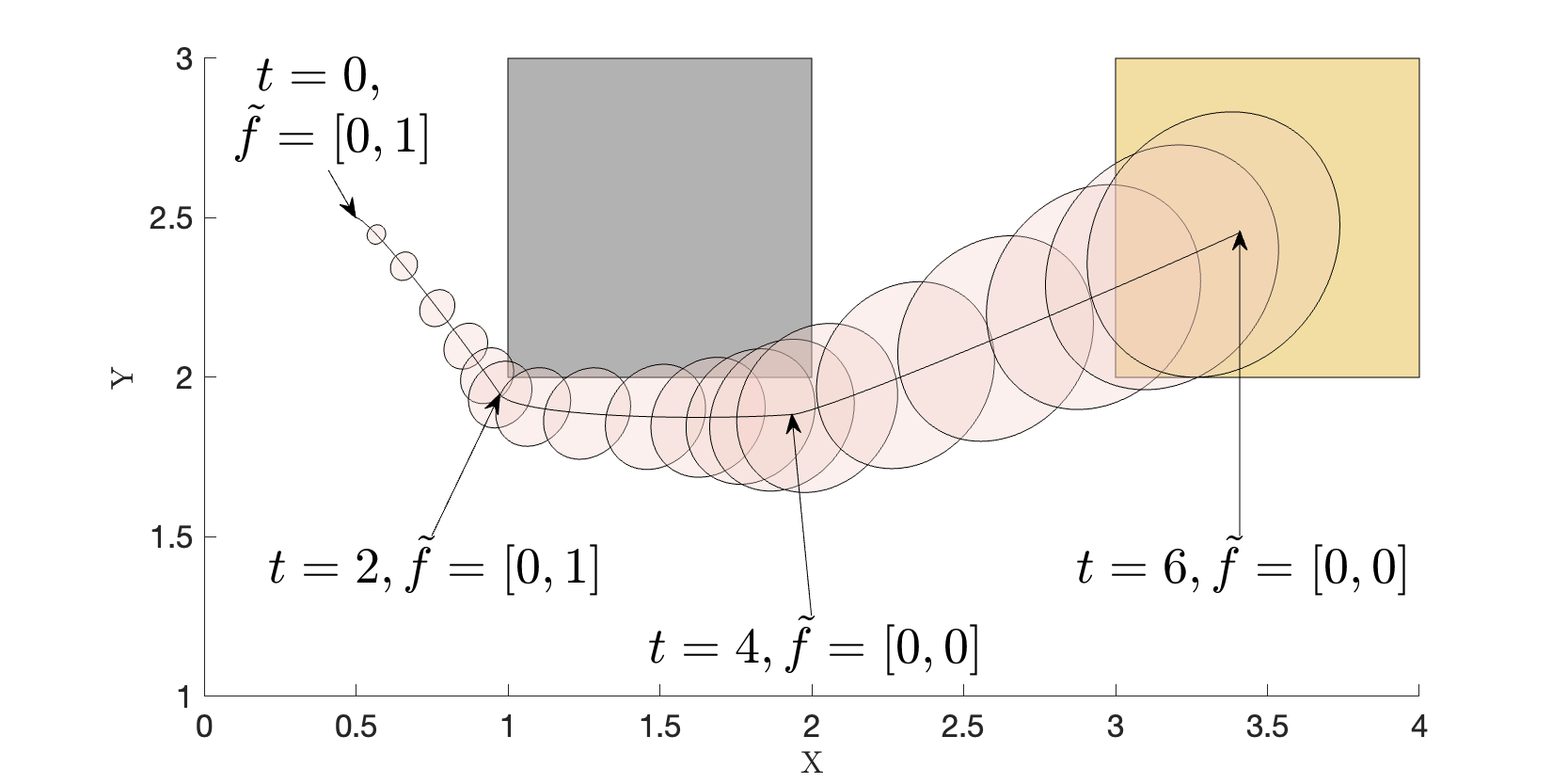}
\caption{Sample trajectories that have a \stori of (left-to-right) $[0.72, 0.91]$, $[0.52, 0.847]$, $[0, 0]$ with respect to $\varphi_1$ in \eqref{eq:form1}}
\label{fig:examples}
\end{figure*}

The \stori for a conjunction of two formulas $\phi_1 \land \phi_2$ is inspired by the lower and upper bounds on the probability of a conjunction of two events. These are the Boole-Fr\'echet inequalities for logical conjunction \cite{frechet1935generalisation}, which make no assumptions on the independence on the events, and hence are general. Note that the measure seeks to quantify both upper and lower bounds because the upper bound is conservative when the specification seeks to enforce `avoidance' of the conjunction, but the lower bound is conservative when it seeks to enforce the `occupancy' of the conjunction.

The \stori of $\phi_1 U_{\langle a,b \rangle} \phi_2$ is less straightforward. Intuitively, it aims to report the 
\storm
at the time $t \in \langle a,b \rangle$ that best balances the point-wise probabilities that $\phi_1$ hold for all $t' \in [0,t]$ and the probability that $\phi_2$ hold at $t$. 

Note that by following Definition \ref{def:stori}, the obtained \stori for eventually and globally are: 
\begin{align}
    f(\Diamond_{\langle a,b\rangle} \phi, \btraj) &= \bigg[\max_{t \in \langle a,b \rangle} \low{f}(\phi,\btraj^t), \; \max_{t \in \langle a,b\rangle} \up{f}(\phi,\btraj^t)\bigg] 
\end{align}
\begin{align}
    f(\square_{\langle a,b\rangle} \phi, \btraj) &= \bigg[\min_{t \in \langle a, b\rangle} \low{f}(\phi,\btraj^t), \; \min_{t \in \langle a,b\rangle} \up{f}(\phi,\btraj^t)\bigg]  
\end{align}

\noindent
These intervals are fully inline with the intuition of the measure as discussed in Example~\ref{ex:intuition}. 

\begin{example}
    \label{ex:stori}
    \emph{
    Consider the environment in Fig. \ref{fig:examples} and formula
    \begin{equation}
    \label{eq:form1}
        \varphi_1 = \neg (x\geq 1 \wedge y \geq 2 \wedge x \leq 2) U_I (x \geq 3 \wedge y \geq 2).
    \end{equation}
    Fig. \ref{fig:examples} shows three belief trajectories, where the ellipses represent the 90\% confidence bounds of their uncertainty. 
    The first two trajectories (left and middle) have the same expectation but are subject to different amounts of process noise. 
    When evaluating against $\varphi_1$ with time interval $I = [0,6]$,
    we see both bounds of the \stori decrease
    from $[0.72, 0.91]$ for the first trajectory (left) to $[0.52, 0.85]$ for the
    the second trajectory that has more uncertainty (middle). This trajectory is less likely to clear the obstacle, and also less likely to arrive in the goal region. The third plot (right) shows the first trajectory but evaluated against $\varphi_1$ with smaller interval $I = [0, 4]$. The robot does not arrive to goal until after 5 seconds. Hence, the robot fails to satisfy the specification and its \stori is $[0, 0]$.
    }
\end{example}


An important property of the \stori is that a trajectory which has a \storm of 1 satisfies STL formula $\phi$ with probability 1, as stated in the following theorem, which illustrates soundness of \storm. 


\begin{theorem}
Given belief trajectory $\btraj$ and STL formula $\phi$, if the Stochastic Robustness Measure $\low{f}(\phi,\btraj) = 1$, then every realization $\xtraj$ of $\btraj$ satisfies the specification, i.e., $\xtraj \models \phi$, except realizations on the set of measure zero.
\end{theorem}
\begin{proof}[Proof]
Consider any realization $\xtraj$ of $\btraj$. If $\low{f}(\phi,\btraj) = 1 \Longrightarrow (\xtraj,0) \models \phi$ and $\up{f}(\phi,\btraj) = 0 \Longrightarrow (\xtraj,0) \not\models \phi$, the \stori follows STL semantics as defined in Definition \ref{def:STLSemant}, which defines the satisfaction of a deterministic trajectory (realizations). This can be shown recursively as follows.

For boolean $\top$, the \stori of $\btraj$ trivially follows STL semantics:
\begin{align*}
    \low{f}(\top,\btraj) = 1 &\Longrightarrow (\xtraj,0) \models \top, \\
\end{align*}
For predicate $\mu$, the \stori of $\btraj$ follows STL semantics:
\begin{align*}
    \low{f}(\mu,\btraj) = 1 &\Longleftrightarrow P(h(x(0)) \geq 0 ) = 1\\
    &\Longleftrightarrow h(x(0)) \geq 0\\ 
    &\Longrightarrow (\xtraj,0) \models \mu,
\end{align*}
For boolean negation $\lnot \phi$, the \stori for $\btraj$ follows STL semantics:
\begin{align*}
    \low{f}(\lnot \phi,\btraj) = 1 &\Longleftrightarrow 1- \up{f}(\phi,\btraj) = 1 \\
    &\Longleftrightarrow \up{f}(\phi,\btraj) = 0 \\
    &\Longrightarrow (\xtraj,0) \not \models \phi \\
    &\Longrightarrow (\xtraj,0) \models \lnot \phi,
\end{align*}
For boolean conjunction $\phi_1 \land \phi_2$, the \stori follows STL semantics:
\begin{align*}
    \low{f}(\phi_1 \land \phi_2,\btraj) = 1 &\Longrightarrow \up{f}(\phi_1 \land \phi_2,\btraj) = 1\\
    &\Longleftrightarrow \min \big\{\low{f}(\phi_1,\btraj),\low{f}(\phi_2,\btraj)\big\} = 1\\
    &\Longleftrightarrow \low{f}(\phi_1,\btraj) = 1 \land \low{f}(\phi_2,\btraj) = 1 \\
    &\Longrightarrow (\xtraj,0) \models \phi_1 \land (\xtraj,0) \models \phi_2 \\
    &\Longrightarrow (\xtraj,0) \models (\phi_1 \land \phi_2),
\end{align*}
and finally, for temporal Until $\phi_1 U_I \phi_2$, the \stori for $\btraj$ reduces to STL semantics:
\begin{align*}
    \low{f}(\phi_1 U_I \phi_2,\btraj) = 1 &\Longrightarrow \up{f}(\phi_1 U_I \phi_2,\btraj) = 1 \\
    &\Longleftrightarrow \max_{t \in I}\Big\{\min\{\up{f}(\phi_2,\btraj^t), \\
    & \hspace{1.5cm} \min_{t' \in [0,t]}\up{f}(\phi_1,\btraj^{t'})\}\Big\} = 1 \\
    &\Longrightarrow \exists t \in I \,s.t. \min\big\{\up{f}(\phi_2,\btraj^t), \\
    & \hspace{1.5cm} \min_{t' \in [0,t]}\up{f}(\phi_1,\btraj^{t'})\big\} = 1 \\
    &\Longleftrightarrow \exists t \in I \,s.t. \Big(\up{f}(\phi_2,\btraj^t) = 1 \\
    & \hspace{1.5cm} \land \min_{t' \in [0,t]}\up{f}(\phi_1,\btraj^{t'}) = 1 \Big)\\
    &\Longrightarrow \exists t \in I \,s.t. (\xtraj,t) \models \phi_2 \\
    & \hspace{1.5cm} \land \forall t' \in [0,t]\up{f}(\phi_1,\btraj^{t'}) = 1 \\
    &\Longrightarrow \exists t \in I \,s.t. (\xtraj,t) \models \phi_2 \\
    & \hspace{1.5cm} \land \forall t' \in [0,t](\xtraj,t') \models \phi_1 \\
    &\Longrightarrow (\xtraj,0) \models (\phi_1 U_I \phi_2)
\end{align*}
Therefore, a belief trajectory $\btraj$ having a \storm of 1 with respect to STL formula $\phi$ implies that every non-zero measure realization $\xtraj$ of $\btraj$ satisfies the STL specification. 
\end{proof}

\subsection{\stori Monitor}


\stori is defined for a given belief trajectory. It does not account for \emph{what could happen if we extend the trajectory.}
However, in some cases (such as planning), we are interested in extending trajectories to achieve a higher \storm. Hence, we need a monitor for \stori, which assumes the given trajectory is to be extended. The monitor must account for all the possible suffixes of the trajectory and how they might change the \storm. In this section, we present a monitor for the \stori with respect to a partial belief trajectory, that acts to bound the achievable \stori of any extensions of it.

\begin{definition}[StoRI Monitor]
    \label{def:Monitor}
    The Stochastic Robustness Interval (\stori) Monitor of a partial belief trajectory $\ptraj$ over time window $[0,t]$ with respect to an STL formula $\phi$ is a functional $\Tilde{f}(\phi, \ptraj)$:
    \begin{equation*}
        \Tilde{f}(\phi, \ptraj) = [\low{\Tilde{f}}(\phi, \ptraj), \up{\Tilde{f}}(\phi, \ptraj)]
    \end{equation*}
    such that $0 \leq \low{\Tilde{f}}(\phi,\ptraj) \leq \up{\Tilde{f}} (\phi,\ptraj) \leq 1$,
    where
    \begin{align*}
        &\Tilde{f}(\top,\ptraj) = f(\top,\ptraj) \\
        &\Tilde{f}(\mu,\ptraj) = f(\mu,\ptraj) \\
        &\low{\Tilde{f}}(\lnot \phi, \ptraj) = 1 - \up{\Tilde{f}}(\phi,\ptraj), \\
        &\up{\Tilde{f}}(\lnot \phi, \ptraj) = 1 - \low{\Tilde{f}}(\phi,\ptraj), \\
        &\low{\Tilde{f}}(\phi_1 \land \phi_2, \ptraj, t) = \max\big\{\low{\Tilde{f}}(\phi_1,\ptraj) + \low{\Tilde{f}}(\phi_2,\ptraj) -1,0\big\}, \\
        &\up{\Tilde{f}}(\phi_1 \land \phi_2, \ptraj, t) = \min\big\{\up{\Tilde{f}}(\phi_1,\ptraj), \up{\Tilde{f}}(\phi_2,\ptraj)\}, \\
        &\low{\Tilde{f}}(\phi_1 U_{\langle a,b \rangle} \phi_2, \ptraj) = 
        \begin{cases}
             \low{f}(\phi_1 U_{\langle a,b \rangle} \phi_2, \ptraj) & \text{if } t \geq a \\ 
            0 &\text{otherwise,}
        \end{cases} \\
        &\up{\Tilde{f}}(\phi_1 U_{\langle a,b \rangle} \phi_2, \ptraj) = \\ 
        &
        \begin{cases}
            \up{f}(\phi_1 U_{\langle a,b \rangle} \phi_2, \ptraj)
            & \hspace{-2mm} \text{if } t \geq b\\
            \displaystyle \max \hspace{-.5mm} \Big\{\hspace{-.5mm}\up{f}(\phi_1 U_{\langle a,b \rangle} \phi_2, \ptraj), \min_{t'\in[0,b]}\up{f}(\phi_1,\ptraj^{t'}) \hspace{-.5mm} \Big\} & \hspace{-2mm} \text{if } t \in [0,b) \\ 
            1 & \hspace{-2mm} \text{otherwise.}
        \end{cases}
    \end{align*}
\end{definition}

The differences between the \stori and \stori Monitor arise in the temporal operators. These operators seek to bound possible future robustness, and also quantify the robustness of the behavior already seen. This is apparent in the \stori Monitor with respect to the $\Diamond$ and $\square$ operators:

\begin{align*}
    \Tilde{f}(\Diamond_{\langle a,b \rangle} \phi, \ptraj) =  
    &
    \begin{cases}
        f(\Diamond_{\langle a,b \rangle} \phi, \ptraj) &\text{if } t \geq b \\
        \Big[\max\limits_{t \in \langle a,b \rangle} \low{f}(\phi,\ptraj^t), \;1\Big]& \text{if } a \leq t < b \\
        \big[0,1\big]& \text{if } t < a
    \end{cases} \\
    \Tilde{f}(\square_{\langle a,b \rangle } \phi, \ptraj) = 
    &
    \begin{cases}
        f(\square_{\langle a,b \rangle } \phi, \ptraj) &\text{if } t \geq b \\
        \Big[0, \; \min\limits_{t \in \langle a,b \rangle} \up{f}(\phi,\ptraj^t)\Big]& \text{if } a \leq t < b \\
        \big[0,1\big]& \text{if } t < a
    \end{cases} 
\end{align*} 

The case when $a\leq t < b$ is of particular interest. In the case of the $\square$ operator, the \stori is upper bounded by the best point in the partial trajectory, but is lower bounded by zero to account for possible future violation. In the case of the $\Diamond$ operator, the \stori is lower bounded by the best point in the partial trajectory, but is upper bounded by one to account for possible future ``perfect" satisfaction. 

The following theorem proves the correctness of \stori Monitor by showing that it always subsumes \stori.

\begin{lemma}
    \label{lemma:monitorupper}
    Let $\btraj$ be a belief trajectory over time window $[0,T]$ and $\btraj_t$ be a prefix of $\btraj$ where $t \in [0,T]$. Then, given STL formula $\phi$, the upper bound of the StoRI Monitor for $\btraj_t$ is monotonically decreasing with respect to $t$.
\end{lemma}
\begin{proof}
    The proof is provided in the Appendix. 
\end{proof}

\begin{lemma}
    \label{lemma:monitorlower}
    Let $\btraj$ be a belief trajectory over time window $[0,T]$ and $\btraj_t$ be a prefix of $\btraj$ where $t \in [0,T]$. Then, given STL formula $\phi$, the lower bound of the StoRI Monitor for $\btraj_t$ is monotonically increasing with respect to $t$.
\end{lemma}
\begin{proof}
    Without loss of generality, consider $\phi = \psi_1 U_{[a,b]} \psi_2$. From Definition \ref{def:Monitor}, the lower bound of the monitor of $\phi$ is:
    \begin{align*}
        &\low{\Tilde{f}}(\phi_1 U_{\langle a,b \rangle} \phi_2, \ptraj) = \\
        & = \low{f}(\phi_1 U_{\langle a,b \rangle} \phi_2, \ptraj) \\
        & = \max_{t\in\langle a,b \rangle} \Big\{ \max \big\{\low{f}(\phi_2,\ptraj^t) + \min_{t'\in[0,t]}\low{f}(\phi_1,\ptraj^{t'}) - 1,0\big\} \Big\},
    \end{align*}
    for $t \geq a$, which monotonically increases with respect to $t$. For $t \not\geq a$, the monitor takes on a value of 0. Because 0 is the lowest value the \stori Monitor can take, it will never decrease with respect to $t$. Therefore, the lower bound of the monitor monotonically increases with respect to t.
\end{proof}

\begin{lemma}
    \label{lemma:subsume}
    Let $\btraj$ be a belief trajectory over time window $[0,T]$. Then, given STL formula $\phi$, the \stori Monitor of $\btraj$ subsumes the \stori of $\btraj$.
\end{lemma}
\begin{proof}
    We show that for any belief trajectory $\btraj$, the value of the StoRI Monitor's upper bound is greater than or equal to the value of the StoRI's upper bound, i.e.,
    \begin{align*}
        &\up{\Tilde{f}}(\phi_1 U_{\langle a,b \rangle} \phi_2, \btraj) = \\
        &= \max \Big\{\up{f}(\phi_1 U_{\langle a,b \rangle} \phi_2, \btraj), \min_{t'\in[0,b]}\up{f}(\phi_1,\btraj^{t'})\Big\} \\
        &\geq \up{f}(\phi_1 U_{\langle a,b \rangle} \phi_2, \btraj) \hspace{1.8cm} \text{for } t \geq 0 \\
        &\text{and } \up{\Tilde{f}}(\phi_1 U_{\langle a,b \rangle} \phi_2, \btraj) = 1 \\
        &\geq \up{f}(\phi_1 U_{\langle a,b \rangle} \phi_2, \btraj) \hspace{1cm} \text{for } t \not\geq 0,
    \end{align*}
    and the value of the StoRI Monitor's lower bound is less than or equal to the value of the StoRI's lower bound, i.e.,
    \begin{align*}
        &\low{\Tilde{f}}(\phi_1 U_{\langle a,b \rangle} \phi_2, \btraj) = \low{f}(\phi_1 U_{\langle a,b \rangle} \phi_2, \btraj) \text{ for when } t \geq a \\
        &\low{\Tilde{f}}(\phi_1 U_{\langle a,b \rangle} \phi_2, \btraj) = 0 \leq \low{f}(\phi_1 U_{\langle a,b \rangle} \phi_2, \btraj) \text{ for when } t \not\geq a
    \end{align*}
    Therefore, \stori Monitor of $\btraj$ subsumes \stori of $\btraj$.
\end{proof}

\begin{theorem}
    \label{theorem:monitorbounds}
    Let $\btraj$ be a belief trajectory over time window $[0,T]$ and $\btraj_t$ be a prefix of $\btraj$ where $t \in [0,T]$.  Then, given STL formula $\phi$, the StoRI Monitor for $\btraj_t$ subsumes the StoRI of $\btraj$ for all $t \in [0,T]$, i.e., $\forall t \in [0,T]$, $f(\phi,\btraj) \subseteq \Tilde{f}(\phi,\btraj_t)$.
\end{theorem}
\begin{proof}[Proof]
    From Lemma \ref{lemma:subsume}, when $t = T$, the \stori Monitor of $\btraj_t$ subsumes \stori of $\btraj$. From Lemmas \ref{lemma:monitorupper} and \ref{lemma:monitorlower}, for every $t < T$, the \stori Monitor of $\btraj_{t}$ subsumes the \stori Monitor of $\btraj_{T}$. Then it follows that the the \stori monitor of $\btraj_t$ subsumes the \stori of $\btraj$ for all $t \leq T$.
\end{proof}


\begin{example}
    \label{ex:monitor}
    \emph{
    Recall the STL formula in Example~\ref{ex:stori} and sample trajectories in Fig. \ref{fig:examples}.
    The \stori Monitor of the trajectories at different points in their evolution are shown in the plots in Fig. \ref{fig:examples}. Note that they subsume the \stori of the final trajectory and the interval gets smaller with time. 
    }
\end{example}

%% file: ICRA23_sections/Methodology-Extended.tex
\section{Motion Planning Algorithm}
\label{sec:motion-planning}
This section presents two sampling-based motion planners that utilize \stori and the \stori Monitor. The first planner finds solutions that satisfy a given \storm constraint, and the second finds solutions that directly optimize for the \storm. 

\subsection{StoRI-\X}
A kinodynamic sampling-based tree planner \X grows a motion tree in the state space $X$ according to the robot dynamics through sampling and extension procedures. 
We generalize planner \X to \stori-\X (Alg. \ref{alg:StoRIX}) to generate plans for System~\eqref{eq:SDE} that are guaranteed to satisfy a given lower bound $\kappa \in [0,1]$ on \storm with respect to STL formula $\phi$.  

\begin{algorithm}
\caption{\stori-\X$(X,U,\phi,\hat{x}_0,N,\kappa)$}
\label{alg:StoRIX}
$P_0 \gets 0^{n \times n}$\;
$\mathbb{V} \gets \{(\hat{x}_0,P_0)\},\mathbb{E} \gets \emptyset $\;
$\mathbb{G} \gets \{\mathbb{V},\mathbb{E}\}$\;
\For{N iterations}{
    $\hat{x}_{rand},t_{rand},\hat{x}_{near},P_{near} \gets sample(X,\mathbb{V},T)$\;
    $\hat{x}_{new},P_{new} \gets $ Extend($\hat{x}_{near},P_{near},U$)\;
    \If{$\up{\Tilde{f}}(\phi,(\hat{x}_0,P_0)...(\hat{x}_{new},P_{new})) > \kappa$}{
        $\mathbb{V} \gets \mathbb{V} \cup \{(\hat{x}_{new},P_{new})\}$ \;
        $\mathbb{E} \gets \mathbb{E} \cup \{[(\hat{x}_{near},P_{near}),(\hat{x}_{new},P_{new})]\}$\;
        \If{$\low{f}(\phi,(\hat{x}_0,P_0)...(\hat{x}_{new},P_{new})) > \kappa$}{
            \Return $(\hat{x}_0,P_0)...(\hat{x}_{new},P_{new})$\;
        }
    }
}
\Return $\emptyset$ 
\end{algorithm}


The algorithm first initializes the tree $\mathbb{G}$ with the belief of $x(0)$. Each node in this tree is a tuple $(\hat{x},P)$ of the mean and covariance of the distribution that describes the state. At every iteration, the algorithm samples a random state $\hat{x}_{rand}$ and time $t_{rand}$ and computes the nearest existing node $(\hat{x}_{near},P_{near})$. $t_{rand}$ is sampled from the time-horizon of the STL formula as defined in \cite{rosi}, and the nearest node is selected using a distance metric that accounts for both state distance $\|\hat{x}_{rand} - \hat{x}_{near}\|$ and time distance $|t_{rand} - t_{near}|$. Second, a random control input $u \in U$ and time duration $t$ is sampled and propagate the system from $(\hat{x}_{near},P_{near})$ to generate a new belief node $(\hat{x}_{new},P_{new})$. Third, the \stori Monitor of the partial belief trajectory $\ptraj = (\hat{x}_0,P_0)(\hat{x}_1,P_1)\cdots(\hat{x}_{new},P_{new})$ is computed for the formula $\phi$. If the \stori Monitor has an upper bound $\up{\Tilde{f}}(\phi,\ptraj) \leq \kappa$, $\ptraj$ has already violated the STL Specification and the new node is discarded. Otherwise, we add $(\hat{x}_{new},P_{new})$ and the edge $\big((\hat{x}_{near},P_{near}),(\hat{x}_{new},P_{new})\big)$ to $\mathbb{G}$. Finally, $\ptraj$ is a solution if the \storm $\low{f}(\phi,\ptraj) > \kappa$. This process repeats until a solution is found or for a maximum of $N$ iterations.

\begin{theorem}[Probabilistic Completeness]
\label{theorem:complete}
Planner \stori-\X in Alg.~\ref{alg:StoRIX} is probabilistically complete if the underlying planner \X is probabilistically complete.
\end{theorem}
\begin{proof}[Proof]
    \stori-\X mimics the behavior of planner \X, only modifying its validity check. From Theorem \ref{theorem:monitorbounds}, this modification only rejects nodes that are guaranteed to violate the STL Specification $\phi$. Therefore, if a solution exists, \stori-\X will find it with probability $1$ as $N \rightarrow \infty$. 
\end{proof}

\subsection{Asymptotically Optimal StoRI-\X (AO-StoRI-\X)}
Since the \storm allows us to compute a quantitative value of robustness, we can also optimize for the \storm of a belief trajectory in a sampling-based motion planner. This is enabled through the AO-$\mathcal{A}$ meta-algorithm in \cite{ao-rrt}, by repeatedly calling \stori-\X with increasing $\kappa$ bounds by setting $\kappa$ as the \storm of the previous solution at 
each iteration.

\begin{algorithm}
\caption{AO-StoRI-\X($X,U,\phi,\hat{x}_0,N)$)}
\label{alg:frost}
$\kappa \gets 0$\;
\While{$N \neq 0$}{
$\btraj,n \gets $StoRI-\X($X,U,\phi,\hat{x}_0,N,\kappa)$)\;
$\kappa \gets \low{f}(\phi,\btraj)$\; 
$N \gets N - n$\;
}
\Return $\btraj$ if solution found, No Solution otherwise
\end{algorithm}

%% file: ICRA23_sections/Experiments-Extended.tex
\section{Evaluations}
We evaluate efficacy and efficiency of the proposed measure and planners subject to a variety of STL specifications. The algorithms are implemented with \X= RRT \cite{RRT}, i.e., \stori-RRT. All algorithms are implemented in the Open Motion Planning Library (OMPL)~\cite{sucan2012the-open-motion-planning-library}, and computations were performed on 3.9 GHz CPU and 64 GB of RAM. The implementation is readily available on GitHub \cite{storirepo}.

We considered a noisy second-order unicycle system whose stochastic dynamics are given by:
    $dx = v \cos(\theta) dt + dw$, 
    $dy = v \sin(\theta) dt + dw$,
    $d\theta = u_\omega dt + dw$,
    $dv = u_a dt + dw$
where $u_a$ and $u_\omega$ are the acceleration and steering angle inputs. We linearized the dynamics according to the feedback linearization in \cite{unicycle}. We evaluate the \storm and \stori Monitor using sequences of points sampled from the continuous trajectory with $dt = 0.15$ seconds, and propagate dynamics and uncertainty according to the process in \cite{simon2006optimal}. 
The beliefs at the sampled times are described by:
\begin{align}
    \hat{x}_{k+1} &= A'\hat{x}_k + B'u_k,
    \label{eq:DT-mean-dynamics}\\
    P_{k+1} &= A'P_kA'^T + Q'
    \label{eq:DT-covariance-dynamics}
\end{align}
where
\begin{align}
    A' &= 
    \begin{bmatrix}
        1 & 0.15 & 0 & 0 \\
        0 & 1 & 0 & 0 \\
        0 & 0 & 1 & 0.15 \\
        0 & 0 & 0 & 1 \\
    \end{bmatrix} \\
    B' &= 
    \begin{bmatrix}
        0.01125 & 0 \\
        0.15 & 0 \\
        0 & 0.01125 \\
        0 & 0.15 \\
    \end{bmatrix} \\
    Q' &= 1e^{-6}\times 
    \begin{bmatrix}
        10 & 1 & 1 & 1 \\
        1 & 10 & 1 & 1 \\
        1 & 1 & 10 & 1 \\
        1 & 1 & 1 & 10 \\
    \end{bmatrix}
\end{align}

Our $sample$ procedure uses the distance metric $\|\hat{x}_{rand} - \hat{x}_{near}\| + 0.25 \times|t_{rand} - t_{near}|$.

We considered environments in Figs. \ref{fig:examples}, \ref{fig:form2}, and \ref{fig:form4} respectively  with STL formulae $\varphi_1$ in \eqref{eq:form1}, 
\begin{align}
    \varphi_2 &= \lnot O U_{[0,10]} A \land \lnot O U_{[0,10]} B, \label{eq:form2}\\
    &O = x \geq 1.5 \land x \leq 3.5 \land y \geq -0.5 \land y \leq 0.5, \notag\\
    &A = x \geq 2 \land x \leq 3 \land y \geq 1, \notag\\
    &B = x \geq 2 \land x \leq 3 \land y \leq -1, \notag
    \end{align}
    \begin{align}
    \varphi_3 &= (\text{Puddle} \rightarrow \lnot \text{Charge} U_{[0,3]} \text{Carpet}) U _{[0,10]} \text{Charge}, \label{formeq:3}\\
    &\text{Puddle} = y \leq 2.5 \land x \geq 2 \land x \leq 3, \notag\\
    &\text{Charge} = y \geq 2 \land x \geq 4, \notag\\
    &\text{Carpet} = y \leq 1 \land x \geq 4. \notag
\end{align}
Here, $\varphi_2$ requires avoiding (black) region $O$ until both regions $A$ and $B$ are visited within 10 minutes in Fig. \ref{fig:form2}. $\varphi_3$ requires the robot to go to the charger within 10 minutes and also that, if it visits the puddle, it must avoid the charger until it visits the carpet within 3 minutes of visiting the puddle in Fig. \ref{fig:form4}. For all formulas, we also require the robot to remain in the workspace for the duration of the mission. 

\subsection{Case Study 1 - Computation Time vs \stori Threshold}

This case study seeks to analyze the relationship between computation time and the \storm constraint $\kappa$. Specifically, we study how \stori-RRT performs 
for formula $\varphi_2$, in the environment in Fig. \ref{fig:form2}. 
We ran $100$ trials, with a maximum computation time of $300$ seconds for each trial. 
Table \ref{tab:bench1} reports the computation time and success rate for different $\kappa$ values. We see that computation time increases and success rate decreases as the $\kappa$ threshold increases. This is due to the increased difficulty of satisfying the robustness constraint.

\begin{table}[h]
    \centering
    \caption{Benchmarking Results for Case Study 1}
    \begin{tabular}{|c|c|c|}
    \hline
        \storm Threshold $\kappa$ & Computation Time (s) & Success Rate\\ 
        \hline
        0.50 & 60.10 $\pm$ 64.87 & 97\%\\
        0.70 & 71.37 $\pm$ 67.48 & 97\%\\
        0.90 & 83.30 $\pm$ 69.35 & 90\%\\
        0.95 & 92.48 $\pm$ 81.56 & 86\%\\
        \hline
    \end{tabular}
    \label{tab:bench1}
\end{table}

Fig. \ref{fig:form2} shows two sample trajectories for this specification, where the ellipses represent the 90\% confidence
bounds of their uncertainty. Trajectory 1 uses $\kappa = 0.9$ and goes through the wider opening on the left. In contrast, Trajectory 2 uses $\kappa = 0.5$ and finds less robust paths that go through the narrow opening on the right. 

\begin{figure}[t]
\centering
\begin{subfigure}{0.235\textwidth}
    \centering
    \includegraphics[width=\linewidth,trim={1cm 0cm 2cm 0cm},clip]{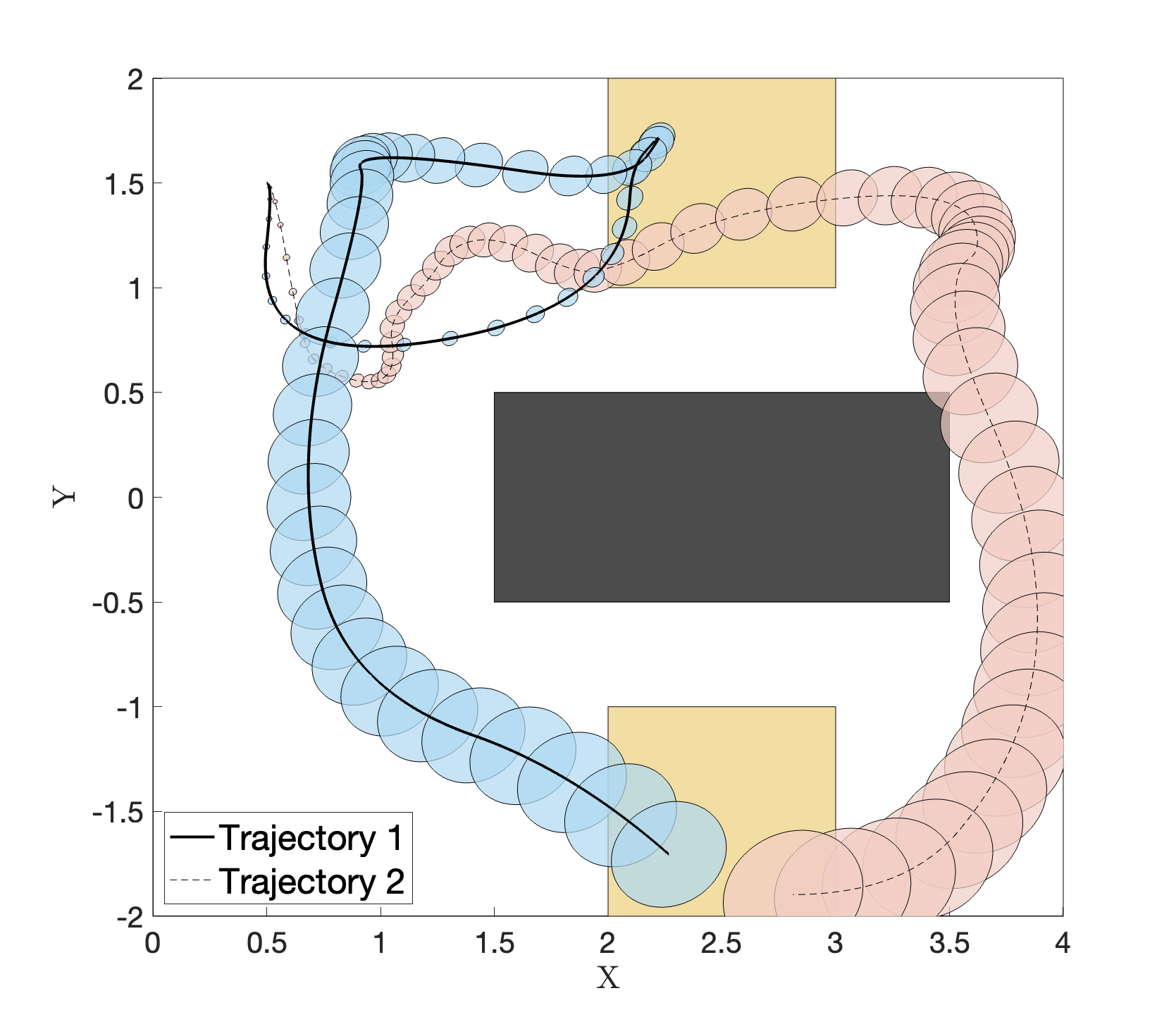}
    \caption{$\phi_2$ Trajectories}
    \label{fig:form2}
\end{subfigure}
\begin{subfigure}{0.235\textwidth}
    \centering
    \includegraphics[width = \linewidth,trim={1cm 0 2cm 0},clip]{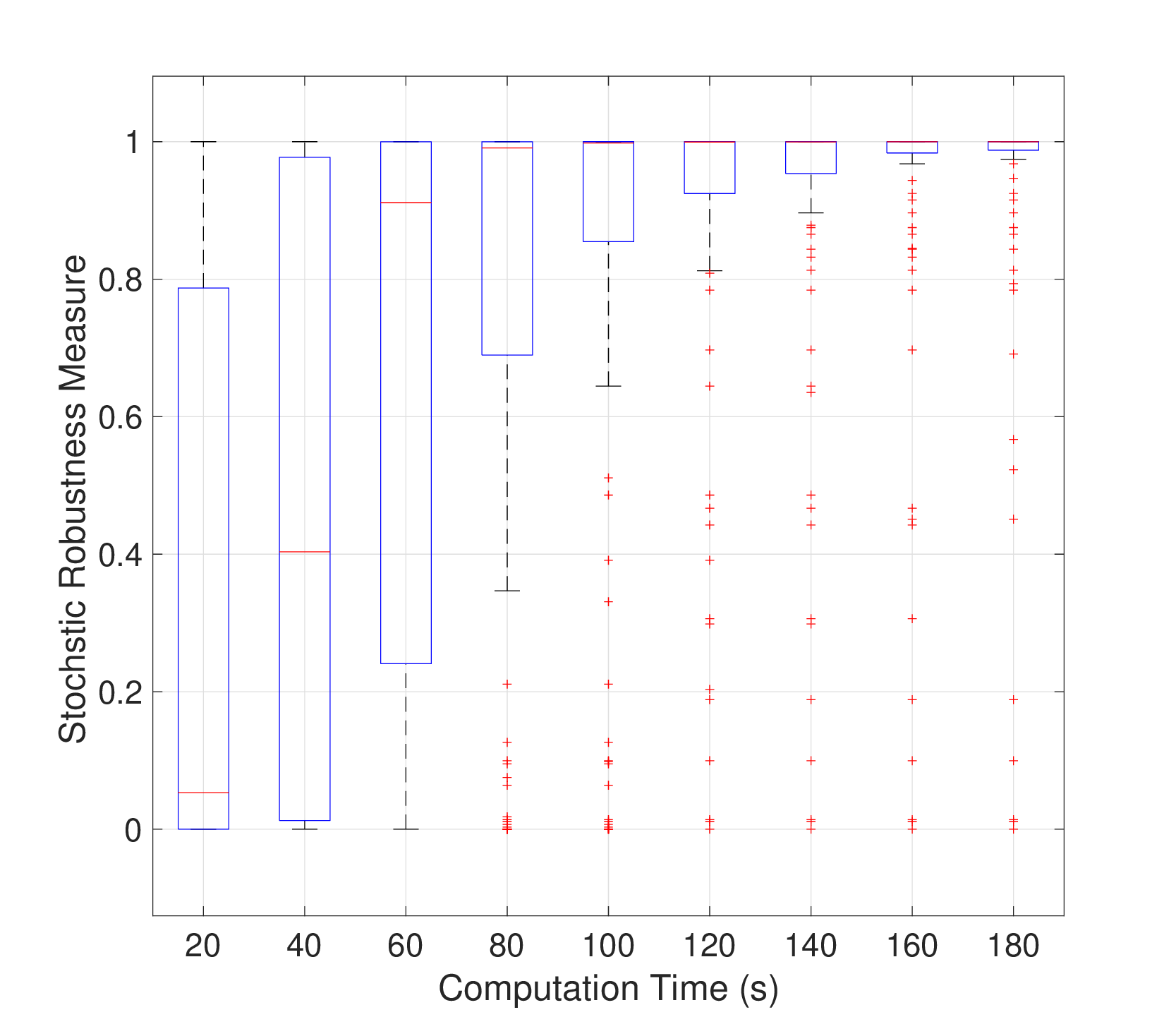}
    \caption{Case Study 3 Benchmark}
    \label{fig:case3}
\end{subfigure}
\caption{Experiment Results. (a) shows two trajectories where $\low{f}(\varphi_2,\btraj_1) = 0.98, \low{f}(\varphi_2,\btraj_2) = 0.67$, and (b) shows benchmarking results for Case Study 3.}
\label{fig:results}
\vspace{-3mm}
\end{figure}

\subsection{Case Study 2 - Computation Time for Different Formulas}
This case study compares computation time for different STL formulas. Table \ref{tab:bench2} shows the average computation time of \stori-RRT for the three formulas and environments over 100 trials with a maximum timeout of 300s. We used \storm bound $\kappa = 0.9$. We see that, even with a tight \storm threshold $\kappa$, the algorithm finds solutions with a good success rate and computation time. This shows that the algorithm is generally applicable to all STL formulas. 

\begin{table}[t]
    \caption{Benchmarking Results for Case Study 2, $\kappa = 0.9$.}
    \label{tab:bench2}
    \centering
    \begin{tabular}{|c|c|c|}
    \hline
        Formula & Computation Time (s) & Success Rate\\ 
        \hline
        $\varphi_1$ & 1.98 $\pm$ 2.56 & 100\%\\
        $\varphi_2$ & 85.95 $\pm$ 82.56 & 90\%\\
        $\varphi_3$ & 44.55 $\pm$ 53.89 & 99\%\\
        \hline
    \end{tabular}
     \vspace{-3mm}
\end{table}

\subsection{Case Study 3 - Asymptotic Optimality}
In this case study, we analyze the relationship between given computation time and the \storm of the resulting trajectory. Here, we study how AO-\stori-RRT performs when planning for formula $\varphi_3$, in the environment in Fig. \ref{fig:form4}. Fig. \ref{fig:case3} presents the results of 100 trials. It clearly shows that the solutions asymptotically approach an optimal \storm. Fig. \ref{fig:form4} gives two sample trajectories with different {\storm}s. Trajectory 2 is found earlier in the optimization, and makes it through the gap. When the planner further optimizes for the \storm, however, it favors trajectories like Trajectory 1, that enter the puddle and dry off before going to the carpet. 

\subsection{Case Study 4 - Simulated Performance}
This case study seeks to analyze how the \storm relates to the \textit{statistical satisfaction rate} (SSR) defined as $\frac{\#\text{satisfying runs}}{\text{total \# runs}}$. We do this by simulating the robot's motion plans from \stori-RRT. Table \ref{tab:corr} compares the \stori of the belief trajectories against the 
SSR
of the simulated realizations. Each motion plan is simulated 1000 times. We use the Breach Toolbox \cite{rosi} to evaluate the realizations' satisfaction. These results show correlation between the \storm  and probability of satisfaction.

\begin{table}[ht]
    \caption{Comparison of StoRI against SSR
    }
    \label{tab:corr}
    \centering
    \begin{tabular}{|c|c|c|c|}
    \hline
        Formula & $\low{f}(\phi,\btraj)$ (\storm) & $\up{f}(\phi,\btraj)$ & Satisfaction Rate\\
        \hline
        $\varphi_1$ & 0.568 & 0.568  & 0.367\\
        $\varphi_1$ & 0.755 & 0.781  & 0.767\\
        $\varphi_1$ & 0.986 & 0.986  & 0.985\\ 
        \hline
        $\varphi_2$ & 0.572 & 0.589  & 0.582\\
        $\varphi_2$ & 0.873 & 0.998  & 0.754\\
        $\varphi_2$ & 0.912 & 0.913  & 0.889\\
        \hline
        $\varphi_3$ & 0.625 & 0.677  & 0.669\\ 
        $\varphi_3$ & 0.876 & 0.878  & 0.872\\
        $\varphi_3$ & 0.989 & 0.994  & 0.985\\
        \hline
    \end{tabular}
     \vspace{-3mm}
\end{table}

%% file: ICRA23_sections/Conclusion.tex
\section{Conclusion and Future Work}
This paper proposes a measure, \storm, for quantifying the robustness of stochastic systems' trajectories with respect to STL specifications. We develop a monitor for this measure that reasons about partial trajectories, and use it in a sampling-based motion planner. We show desirable properties of this measure, that the algorithm is probabilistically complete, and that the algorithm asymptotically optimizes for the \storm. Emprical evaluation demonstrates the measure and algorithm's effectiveness and utility. For future work, we plan to investigate   guiding the growth of the motion tree and incorporating measurement uncertainty in the
 planner. 

%% file: ICRA23_sections/Appendix.tex
\appendix
\label{sec:appendix}
\begin{proof}[Proof of Lemma \ref{lemma:monitorupper}]
    The proof is as follows: Consider the value of the \stori Monitor upper bound of two prefixes $\btraj_{t_1}$ and $\btraj_{t_2}$, where $t_2 \geq t_1$. By definition \ref{def:Monitor}, the value of the upper bound of the \stori Monitor is one of $4$ terms. We show that, for each case, the value of the upper bound of the \stori Monitor of $\btraj_{t_2}$ is always less than or equal to the value of the upper bound of the \stori Monitor of $\btraj_{t_1}$. It follows then, that the upper bound of the \stori Monitor monotonically decreases with respect to $t$. \\
     Without loss of generality, consider $\psi = \phi_1 U_{[a,b]} \phi_2$. From Definition \ref{def:Monitor}, the upper bound of the monitor of $\btraj_{t_1}$ with respect to $\psi$ when $t_1 \in [0,b)$ is:
    \begin{align*}
        &\up{\Tilde{f}}(\phi_1 U_{\langle a,b \rangle} \phi_2, \btraj_{t_1}) = \\
        &= \max \Big\{\up{f}(\phi_1 U_{\langle a,b \rangle} \phi_2, \btraj_{t_1}), \min_{t'\in[0,b]}\up{f}(\phi_1,\btraj_{t_1}^{t'})\Big\} \\
        &= \max \bigg\{\max_{t\in\langle a,b \rangle} \Big\{ \min \big\{\up{f}(\phi_2,\btraj_{t_1}^t), \min_{t'\in[0,t]}\up{f}(\phi_1,\btraj_{t_1}^{t'})\big\}\Big\}, \\
        &\hspace{5.5cm}\min_{t'\in[0,b]}\up{f}(\phi_1,\btraj_{t_1}^{t'})\bigg\} \\
        &= \max \bigg\{\max_{t\in\langle a,t_1 \rangle} \Big\{ \min \big\{\up{f}(\phi_2,\btraj_{t_1}^t), \min_{t'\in[0,t]}\up{f}(\phi_1,\btraj_{t_1}^{t'})\big\}\Big\}, \\
        &\hspace{5.5cm}\min_{t'\in[0,t_1]}\up{f}(\phi_1,\btraj_{t_1}^{t'})\bigg\}
    \end{align*}
    We now examine the case where the second term is greater than the first, the case where the first term is greater than the second, the case where there is not yet any information relevant to this formula, and the case where the time interval is closed ($t_1\geq b$).

    \textit{Case 1:} We first examine the case where the second term is greater, i.e.,
    \begin{align}
        &\min_{t'\in[0,t_1]}\up{f}(\phi_1,\btraj_{t_1}^{t'}) \geq \max_{t\in\langle a,t_1 \rangle} \Big\{ \min \big\{\up{f}(\phi_2,\btraj_{t_1}^t), \notag\\
        & \hspace{5.5cm}\min_{t'\in[0,t]}\up{f}(\phi_1,\btraj_{t_1}^{t'})\big\}\Big\} \notag\\
        &\Longrightarrow \up{\Tilde{f}}(\phi_1 U_{\langle a,b \rangle} \phi_2, \btraj_{t_1}) = \min_{t'\in[0,t_1]}\up{f}(\phi_1,\btraj_{t_1}^{t'}) \label{eq:c11}
    \end{align}
    The upper bound of the monitor of $\btraj_{t_2}$ with respect to $\psi$ is:
    \begin{align}
        &\up{\Tilde{f}}(\phi_1 U_{\langle a,b \rangle} \phi_2, \btraj_{t_2}) = \notag\\
        &= \max \bigg\{\max_{t\in\langle a,t_2 \rangle} \Big\{ \min \big\{\up{f}(\phi_2,\btraj_{t_2}^t), \min_{t'\in[0,t]}\up{f}(\phi_1,\btraj_{t_2}^{t'})\big\}\Big\}, \notag\\
        &\hspace{5cm}\min_{t'\in[0,t_2]}\up{f}(\phi_1,\btraj_{t_2}^{t'})\bigg\} \label{eq:c12}
    \end{align}
    We now examine how both of these terms relate to equation \ref{eq:c11}. 
    
    \textit{Case 1.1:}
    We first look at how the first term in equation \ref{eq:c12} relates to the bound in \ref{eq:c11}. This term can be split as follows:
    \begin{align*}
        & \max_{t\in\langle a,t_2 \rangle} \Big\{ \min \big\{\up{f}(\phi_2,\btraj_{t_2}^t), \min_{t'\in[0,t]}\up{f}(\phi_1,\btraj_{t_2}^{t'})\big\}\Big\} \\
        & = \max \bigg\{ \max_{t\in\langle a,t_1 \rangle} \Big\{ \min \big\{\up{f}(\phi_2,\btraj_{t_2}^t), \min_{t'\in[0,t]}\up{f}(\phi_1,\btraj_{t_2}^{t'})\big\}\Big\}, \\
        &\hspace{1.4cm}\max_{t\in\langle t_1,t_2 \rangle} \Big\{ \min \big\{\up{f}(\phi_2,\btraj_{t_2}^t), \min_{t'\in[0,t]}\up{f}(\phi_1,\btraj_{t_2}^{t'})\big\}\Big\}\bigg\}\\
        &= \max\{S_{1},S_{2}\}
    \end{align*}
    Recall, by our assumptions for case 1, term $S_{1}$ is less than or equal to the value in equation \ref{eq:c11}, i.e.,
    \begin{equation}
        S_{1} \leq \min_{t'\in[0,t_1]}\up{f}(\phi_1,\btraj_{t_1}^{t'})
        \label{eq:c121}
    \end{equation}
    The term $S_{2}$ is upper bounded by:
    \begin{align}
        &\max_{t\in\langle t_1,t_2 \rangle} \Big\{ \min \big\{\up{f}(\phi_2,\btraj_{t_2}^t), \min_{t'\in[0,t]}\up{f}(\phi_1,\btraj_{t_2}^{t'})\big\}\Big\}\bigg\}\notag\\
        &\leq \max_{t\in\langle t_1,t_2 \rangle} \Big\{ \min \big\{1, \min_{t'\in[0,t]}\up{f}(\phi_1,\btraj_{t_2}^{t'})\big\}\Big\}\bigg\}\notag\\
        &= \max_{t\in\langle t_1,t_2 \rangle} \Big\{\min_{t'\in[0,t]}\up{f}(\phi_1,\btraj_{t_2}^{t'})\Big\}\bigg\}\notag\\
        &\leq \min_{t'\in[0,t_1]}\up{f}(\phi_1,\btraj_{t_1}^{t'})\notag\\
        &\Longrightarrow S_{2} \leq \min_{t'\in[0,t_1]}\up{f}(\phi_1,\btraj_{t_1}^{t'}) \label{eq:c122}
    \end{align}
    From equations \ref{eq:c121} and \ref{eq:c122}, we conclude that:
    \begin{align}
        &\max_{t\in\langle a,t_2 \rangle} \Big\{ \min \big\{\up{f}(\phi_2,\btraj_{t_2}^t), \min_{t'\in[0,t]}\up{f}(\phi_1,\btraj_{t_2}^{t'})\big\}\Big\} \notag\\
        &\hspace{5cm}\leq \min_{t'\in[0,t_1]}\up{f}(\phi_1,\btraj_{t_1}^{t'}) \label{eq:c13}
    \end{align}
    
    \textit{Case 1.2:}
    We now look at how the second term in equation \ref{eq:c12} relates to the bound in \ref{eq:c11}. This term can be split as follows:
    \begin{align*}
        &\min_{t'\in[0,t_2]}\up{f}(\phi_1,\btraj_{t_2}^{t'})\\
        &= \min\big\{\min_{t'\in[0,t_1]}\up{f}(\phi_1,\btraj_{t_2}^{t'}),\min_{t'\in[t_1,t_2]}\up{f}(\phi_1,\btraj_{t_2}^{t'})\big\}
    \end{align*}
    And this value is always less than the bound in equation \ref{eq:c11}, i.e.,
    \begin{equation}
        \min_{t'\in[0,t_2]}\up{f}(\phi_1,\btraj_{t_2}^{t'}) \leq \min_{t'\in[0,t_1]}\up{f}(\phi_1,\btraj_{t_2}^{t'})
        \label{eq:c14}
    \end{equation}
    Equations \ref{eq:c13} and \ref{eq:c14} show that, for case 1, the \stori Monitor of $\btraj_{t_2}$ with respect to $\psi$ is less than or equal to the \stori Monitor of $\btraj_{t_1}$:
    \begin{equation}
        \boxed{
        \up{\Tilde{f}}(\phi_1 U_{\langle a,b \rangle} \phi_2, \btraj_{t_2}) \leq \up{\Tilde{f}}(\phi_1 U_{\langle a,b \rangle} \phi_2, \btraj_{t_1}) \text{, case 1}
        }
        \label{eq:case1}
    \end{equation}
    
    \textit{Case 2:}
    We next examine the case where the first term of the \stori Monitor of $\btraj_{t_1}$ with respect to $\psi$ is greater, i.e.,
    \begin{align}
        &\max_{t\in\langle a,t_1 \rangle} \Big\{ \min \big\{\up{f}(\phi_2,\btraj_{t_1}^t), \min_{t'\in[0,t]}\up{f}(\phi_1,\btraj_{t_1}^{t'})\big\}\Big\} \notag\\
        &\geq \min_{t'\in[0,t_1]}\up{f}(\phi_1,\btraj_{t_1}^{t'}) \notag\\
        &\Longrightarrow \up{\Tilde{f}}(\phi_1 U_{\langle a,b \rangle} \phi_2, \btraj_{t_1}) = \notag\\
        &\max_{t\in\langle a,t_1 \rangle} \Big\{ \min \big\{\up{f}(\phi_2,\btraj_{t_1}^t), \min_{t'\in[0,t]}\up{f}(\phi_1,\btraj_{t_1}^{t'})\big\}\Big\} \label{eq:c21}
    \end{align}
    We now examine how this value relates to both terms in the upper bound of the monitor of $\btraj_{t_2}$ with respect to $\psi$ (equation \ref{eq:c12}). 
    
    \textit{Case 2.1:}
    We first look at how the first term in equation \ref{eq:c12} relates to the bound in \ref{eq:c21}. This term is again split into $S_1$ and $S_2$. For case 2, we see that $S_1$ is the same value as equation $\ref{eq:c21}$:
    \begin{equation}
        S_1 = \max_{t\in\langle a,t_1 \rangle} \Big\{ \min \big\{\up{f}(\phi_2,\btraj_{t_1}^t), \min_{t'\in[0,t]}\up{f}(\phi_1,\btraj_{t_1}^{t'})\big\}\Big\}
        \label{eq:c211}
    \end{equation}
    Futhermore, equation \ref{eq:c122} still holds. We know by our assumptions for case 2 that this value is less than or equal to the value in equation \ref{eq:c21}, i.e.,
    \begin{align}
        &S_{2} \leq \min_{t'\in[0,t_1]}\up{f}(\phi_1,\btraj_{t_1}^{t'}) \notag\\
        &\leq \max_{t\in\langle a,t_1 \rangle} \Big\{ \min \big\{\up{f}(\phi_2,\btraj_{t_1}^t), \min_{t'\in[0,t]}\up{f}(\phi_1,\btraj_{t_1}^{t'})\big\}\Big\} \notag\\
        &\Longrightarrow S_2 \leq \max_{t\in\langle a,t_1 \rangle} \Big\{ \min \big\{\up{f}(\phi_2,\btraj_{t_1}^t), \min_{t'\in[0,t]}\up{f}(\phi_1,\btraj_{t_1}^{t'})\big\}\Big\} \label{eq:c222}
    \end{align}
    From equations \ref{eq:c211} and \ref{eq:c222}, we conclude that:
    \begin{align}
        &\max_{t\in\langle a,t_2 \rangle} \Big\{ \min \big\{\up{f}(\phi_2,\btraj_{t_2}^t), \min_{t'\in[0,t]}\up{f}(\phi_1,\btraj_{t_2}^{t'})\big\}\Big\} \notag\\
        &\leq \max_{t\in\langle a,t_1 \rangle} \Big\{ \min \big\{\up{f}(\phi_2,\btraj_{t_2}^t), \min_{t'\in[0,t]}\up{f}(\phi_1,\btraj_{t_2}^{t'})\big\}\Big\} \label{eq:c23}
    \end{align}
    
    \textit{Case 2.2:}
    We now look at how the second term in equation \ref{eq:c12} relates to the bound in equation \ref{eq:c21}. Equation \ref{eq:c14} still holds. Furthermore, We know by our assumptions that this value is less than or equal to the value in equation \ref{eq:c21}, i.e.,
    \begin{align}
        &\min_{t'\in[0,t_2]}\up{f}(\phi_1,\btraj_{t_2}^{t'}) \leq \min_{t'\in[0,t_1]}\up{f}(\phi_1,\btraj_{t_2}^{t'}) \\
        &\leq \max_{t\in\langle a,t_1 \rangle} \Big\{ \min \big\{\up{f}(\phi_2,\btraj_{t_2}^t), \min_{t'\in[0,t]}\up{f}(\phi_1,\btraj_{t_2}^{t'})\big\}\Big\} \label{eq:c24}
    \end{align}
    Equations \ref{eq:c23} and \ref{eq:c24} show that, for case 2, the \stori Monitor of $\btraj_{t_2}$ is less than or equal to the \stori Monitor of $\btraj_{t_1}$:
    \begin{equation}
        \boxed{
        \up{\Tilde{f}}(\phi_1 U_{\langle a,b \rangle} \phi_2, \btraj_{t_2}) \leq \up{\Tilde{f}}(\phi_1 U_{\langle a,b \rangle} \phi_2, \btraj_{t_1}) \text{, case 2}
        }
        \label{eq:case2}
    \end{equation}
    
    \textit{Case 3:}
    The third case is when $t < 0$. This could happen when the $U_I$ operator is nested within another temporal operator, resulting in the time-shifted trajectory having a negative time. It means that there are no points in the trajectory relevant to the operator yet. Definition \ref{def:Monitor} defines the upper bound as 1 for this case. The Monitor itself is defined as having an upper bound less than or equal to 1; the value of the monitor cannot increase past 1. Therefore, for case 3, the \stori Monitor of $\btraj_{t_2}$ is less than or equal to the \stori Monitor of $\btraj_{t_1}$:
    \begin{equation}
        \boxed{
        \up{\Tilde{f}}(\phi_1 U_{\langle a,b \rangle} \phi_2, \btraj_{t_2}) \leq \up{\Tilde{f}}(\phi_1 U_{\langle a,b \rangle} \phi_2, \btraj_{t_1}) \text{, case 3}
        }
        \label{eq:case3}
    \end{equation}

    \textit{Case 4:}
    The fourth case is when $t_1 \geq b$. Beyond this point, the \stori Monitor will not change, and has converged to the value of the \stori:
    \begin{align*}
        &\up{\Tilde{f}}(\phi_1 U_{\langle a,b \rangle} \phi_2, \btraj_{t_1}) = \up{f}(\phi_1 U_{\langle a,b \rangle} \phi_2, \btraj_{t_1}) \\
        &= \max_{t\in\langle a,b \rangle} \Big\{ \min \big\{\up{f}(\phi_2,\btraj_{t_1}^t), \min_{t'\in[0,t]}\up{f}(\phi_1,\btraj_{t_1}^{t'})\big\}\Big\}, \\
    \end{align*}

    Because $t_2 \geq t_1$, the time $t\in\langle a,b \rangle$ that maximizes the value above will be the same value for both trajectories:
    \begin{equation}
        \boxed{
        \up{\Tilde{f}}(\phi_1 U_{\langle a,b \rangle} \phi_2, \btraj_{t_2}) = \up{\Tilde{f}}(\phi_1 U_{\langle a,b \rangle} \phi_2, \btraj_{t_1}) \text{, case 4}
        }
        \label{eq:case4}
    \end{equation}
    
    By  \eqref{eq:case1}, \eqref{eq:case2}, \eqref{eq:case3}, and \eqref{eq:case4}, the \stori Monitor of $\btraj_{t_2}$ is always less than or equal to the \stori Monitor of $\btraj_{t_1}$. Therefore, the upper bound of the monitor monotonically decreases with respect to $t$.
\end{proof}